\documentclass{article}
\usepackage[cmex10]{amsmath}
\usepackage{amsthm}

\usepackage[paperwidth=210mm,
            paperheight=297mm,
            text={150mm,230mm},
            left=30mm,
            includeheadfoot,
            vmarginratio=24:19
            ]{geometry}
\newtheorem{lemma}{Lemma}

\usepackage[noadjust]{cite}

\begin{document}
%
\title{The equivalent identities of the MacWilliams identity for linear codes}
\author{Xiaomin Bao\\
       School of Mathematics and Statistics, \\
       Southwest University, \\
       Chongqing, 400715, P.R.China\\
       xbao@swu.edu.cn}


\maketitle

\begin{abstract}
We use derivatives to prove the equivalences between MacWilliams identity and its four equivalent forms,
and present new interpretations for the four equivalent forms. Our results explicitly give out the relationships between
MacWilliams identity and its four equivalent forms.
\end{abstract}

{\bf Keywords} Linear code, MacWilliams identity, equivalent, derivative.

%

\section{Introduction}
Let $\mathcal{C}$ be a $(n,k)$ linear code on the
field $F_q = GF(q)$ and let $\mathcal{C}^{\perp}$ be its dual code.
Define
$$W_{\mathcal{C}}^i:= \mbox{the number of codewords of weight $i$ in $\mathcal{C}$}$$
The homogeneous polynomial
\[
W_{\mathcal{C}}(x,y) := W_{\mathcal{C}}^0y^n + W_{\mathcal{C}}^1xy^{n-1} + \cdots + W_{\mathcal{C}}^nx^n
    = \sum_{i=0}^nW_{\mathcal{C}}^ix^iy^{n-i}
\]
is called weight enumerator of the code $\mathcal{C}$. The following identity is called the MacWilliams identity:
\begin{equation}\label{eq:macIdentity}
   W_{\mathcal{C}}(x,y):= \frac{1}{q^{n-k}}W_{\mathcal{C}^{\perp}}\left(x + (q-1)y,x - y\right)
\end{equation}
The following are four equivalent forms of the MacWilliams identity:
\begin{equation}\label{eq:Aj&Bj}
    W_{\mathcal{C}}^r = \frac{1}{q^{n-k}}\sum_{j=0}^n W_{\mathcal{C}^{\perp}}^j\sum_{i=0}^r(-1)^i\binom{n-j}{r-i}\binom{j}{i}(q-1)^{r-i},
     \quad r = 0,1,\cdots,n
\end{equation}
\begin{equation}\label{eq:NoOfNoneZeroR-Tuple}
    \sum_{j=0}^n\binom{j}{r}W_{\mathcal{C}}^j = q^{k-r}\sum_{j=0}^n(-1)^j(q-1)^{r-j}\binom{n-j}{r-j}W_{\mathcal{C}^{\perp}}^j,
                                              \quad r = 0,1,\cdots,n
\end{equation}
\begin{equation}\label{eq:NoOfZeroR-Tuple}
    \sum_{j=0}^n\binom{n-j}{r}W_{\mathcal{C}}^j = q^{k-r}\sum_{j=0}^n\binom{n-j}{r-j}W_{\mathcal{C}^{\perp}}^j,
                                                \quad r = 0,1,\cdots,n
\end{equation}
\begin{equation}\label{eq:general}
\sum_{j=0}^n\binom{j}{t}\binom{n-j}{r-t}W_{\mathcal{C}}^j
=q^{k-r}\sum_{i=0}^t(-1)^i(q-1)^{t-i}\sum_{j=0}^r\binom{n-j}{r-j}\binom{j}{i}\binom{r-j}{t-i}W_{\mathcal{C}^{\perp}}^j,
0\leq t\leq r\leq n
\end{equation}
The MacWilliams identities and the four equivalent forms have been studied by many authors
\cite{blahut,brualdi,goldwasser,inductionproof,mid,pless,probabilityproof,zierler}.
In 1963, MacWilliams~\cite{mid} proved that \eqref{eq:Aj&Bj}, \eqref{eq:NoOfNoneZeroR-Tuple}
and \eqref{eq:NoOfZeroR-Tuple} are all equivalent to ~MacWilliams identity~\eqref{eq:macIdentity}.
In 1983, by using a method different from that of \cite{mid}, Blahut~\cite{blahut} proved that
\eqref{eq:macIdentity} can be derived from ~\eqref{eq:NoOfZeroR-Tuple}. Similar method can also be used to derive
\eqref{eq:macIdentity} from \eqref{eq:NoOfNoneZeroR-Tuple}.
Identity \eqref{eq:general} was initially discovered by Brualdi et al in 1980 \cite{brualdi}, and they showed that \eqref{eq:general} can be derived  from~\eqref{eq:Aj&Bj}. In 1997,
Goldwasser~\cite{goldwasser} proved~\eqref{eq:general} by induction.

It should be pointed out that Brualdi et al presented interesting combinatorial
interpretations for \eqref{eq:NoOfNoneZeroR-Tuple},
\eqref{eq:NoOfZeroR-Tuple} and \eqref{eq:general} in~\cite{brualdi},
but the interpretations do not indicate any explicit relationship between \eqref{eq:NoOfNoneZeroR-Tuple},
\eqref{eq:NoOfZeroR-Tuple}, \eqref{eq:general} and \eqref{eq:macIdentity}.

In the following section we will use derivatives to prove the
equivalence between anyone of \eqref{eq:Aj&Bj},
\eqref{eq:NoOfNoneZeroR-Tuple}, \eqref{eq:NoOfZeroR-Tuple},
\eqref{eq:general} and \eqref{eq:macIdentity}, our proofs also
unveil new relationships between MacWilliams identity and its equivalent forms.
\section{Proofs of equivalences}
The following two lemmas are needed in our equivalence proofs:
\begin{lemma}\label{lemma:1}
Let $X = x+(q-1)y$,$Y = x-y$,$f = X^sY^t$, then for any non-negative
integers $l,m$ we have
\[
\begin{aligned}
\frac{\partial^{l}f}{\partial x^{l}} &= \sum_{i=0}^ll!\binom{s}{l-i}\binom{t}{i}X^{s-l+i}Y^{t-i}\\
\frac{\partial^{m}f}{\partial y^{m}} &=
\sum_{i=0}^m(-1)^i(q-1)^{m-i}m!\binom{s}{m-i}\binom{t}{i}X^{s-m+i}Y^{t-i}
\end{aligned}
\]
\end{lemma}
\begin{lemma}\label{lemma:2}
Let $f(x,y)$ and $g(x,y)$ be two homogeneous polynomials of degree
$n$ in $x,y$. If
$$\left.\frac{\partial^rf}{\partial y^r}\right|_{x=1,y=0} = \left.\frac{\partial^rg}{\partial y^r}\right|_{x=1,y=0},\quad 0\leq r\leq n$$
or
$$\left.\frac{\partial^rf}{\partial y^r}\right|_{x=0,y=1} = \left.\frac{\partial^rg}{\partial y^r}\right|_{x=0,y=1},\quad 0\leq r\leq n$$
or
$$\left.\frac{\partial^rf}{\partial y^r}\right|_{x=y=1} = \left.\frac{\partial^rg}{\partial y^r}\right|_{x=y=1},\quad 0\leq r\leq n$$
then $f(x,y)=g(x,y)$.
\end{lemma}
\begin{proof}[Proof of Lemma \ref{lemma:1}]
We only prove the second identity, the first one can be proved
similarly.

If $m=0$, the result is obvious. Now let $m>0$, and suppose
\[
\frac{\partial^{m-1}f}{\partial y^{m-1}} = \sum_{i=0}^{m-1}(-1)^i(q-1)^{m-1-i}(m-1)!
                                         \binom{s}{m-1-i}\binom{t}{i}X^{s-m+1+i}Y^{t-i}
\]
Then from $\frac{\partial^{m}f}{\partial y^{m}} =
\frac{\partial(\partial^{m-1}f/\partial y^{m-1})}{\partial y}$ we
can get
\[
\begin{aligned}
  \frac{\partial^{m}f}{\partial y^{m}} &= \sum_{i=0}^{m-1}(-1)^i(q-1)^{m-i}(s-m+1+i)(m-1)!
                                       \binom{s}{m-1-i}\binom{t}{i}X^{s-m+i}Y^{t-i}\\
                                       &\qquad +\sum_{i=0}^{m-1}(-1)^{i+1}(t-i)(q-1)^{m-1-i}(m-1)!
                                       \binom{s}{m-1-i}\binom{t}{i}X^{s-m+1+i}Y^{t-i-1}\\
                                       &=\sum_{i=0}^{m-1}(-1)^i(q-1)^{m-i}(m-1)!(m-i)
                                       \binom{s}{m-i}\binom{t}{i}X^{s-m+i}Y^{t-i}\\
                                       & \qquad +\sum_{i=0}^{m-1}(-1)^{i+1}(q-1)^{m-(i+1)}(m-1)!(i+1)
                                       \binom{s}{m-(i+1)}\binom{t}{i+1}X^{s-m+(i+1)}Y^{t-(i+1)}\\
                                       &=\sum_{i=0}^{m-1}(-1)^i(q-1)^{m-i}(m-1)!(m-i)
                                       \binom{s}{m-i}\binom{t}{i}X^{s-m+i}Y^{t-i}\\
                                       &\qquad +\sum_{i=1}^{m}(-1)^{i}(q-1)^{m-i}(m-1)!\,i
                                       \binom{s}{m-i}\binom{t}{i}X^{s-m+i}Y^{t-i}\\
                                       &=\sum_{i=0}^m(-1)^i(q-1)^{m-i}m!\binom{s}{m-i}\binom{t}{i}X^{s-m+i}Y^{t-i}
\end{aligned}
\]
The assertion follows by induction.
\end{proof}
\begin{proof}[Proof of Lemma \ref{lemma:2}]
We only prove the case of
\begin{equation}\label{eq:partialDir}
\left.\frac{\partial^rf}{\partial y^r}\right|_{x=y=1} =
\left.\frac{\partial^rg}{\partial y^r}\right|_{x=y=1},\quad 0\leq
r\leq n
\end{equation}
the other two cases can be proved similarly.

Let
$$f(x,y) = \sum_{i=0}^nf_{i}x^{n-i}y^i,\quad g(x,y) = \sum_{i=0}^ng_{i}x^{n-i}y^i$$
then from~\eqref{eq:partialDir} we can get the following equations:
\[
\begin{aligned}
    n!f_n &= n!g_n\\
    (n-1)!\sum_{i=n-1}^n\binom{i}{n-1}f_{i}&= (n-1)!\sum_{i=n-1}^n\binom{i}{n-1}g_{i}\\
    (n-2)!\sum_{i=n-2}^n\binom{i}{n-2}f_{i}&= (n-2)!\sum_{i=n-2}^n\binom{i}{n-2}g_{i}\\
    \vdots\qquad  & \qquad \vdots\\
    2!\sum_{i=2}^n\binom{i}{2}f_{i}&= 2!\sum_{i=2}^n\binom{i}{2}g_{i}\\
    \sum_{i=1}^n\binom{i}{1}f_{i}&= !\sum_{i=1}^n\binom{i}{1}g_{i}
\end{aligned}
\]
Solving these equations we get
$$f_n=g_n,f_{n-1}=g_{n-1},\cdots, f_1 = g_1,f_0=g_0$$
Therefore $f(x,y)=g(x,y)$.
\end{proof}
\subsection{Derive \eqref{eq:Aj&Bj} or \eqref{eq:NoOfNoneZeroR-Tuple} from \eqref{eq:macIdentity}}
By taking $r$-th partial derivative with respect to $y$ on both sides of \eqref{eq:macIdentity}, we get
\[
\begin{aligned}
 \sum_{j=0}^n r!\binom{j}{r}W_{\mathcal{C}}^jx^{n-j}y^{j-r}
 &= \frac{1}{q^{n-k}}\sum_{j=0}^n W_{\mathcal{C}^{\perp}}^j\sum_{i=0}^r(-1)^ir!\binom{n-j}{r-i}\binom{j}{i}
  (q-1)^{r-i}\left[x+(q-1)y\right]^{n-j-r+i}(x-y)^{j-i}
\end{aligned}
\]
\begin{itemize}
  \item Substituting $1$ for $x$, $0$ for $y$ in the above equation we get
  $$W_{\mathcal{C}}^r = \frac{1}{q^{n-k}}\sum_{j=0}^n W_{\mathcal{C}^{\perp}}^j\sum_{i=0}^r(-1)^i\binom{n-j}{r-i}\binom{j}{i}(q-1)^{r-i}$$
  So from \eqref{eq:macIdentity} we can derive \eqref{eq:Aj&Bj}.
  \item Substituting $1$ for both $x$ and $y$ we get
\[
\begin{aligned}
  \sum_{j=0}^n \binom{j}{r}W_{\mathcal{C}}^j&= \sum_{j=r}^n \binom{j}{r}W_{\mathcal{C}}^j\qquad (\mbox{if}\;j<r\;\mbox{then}\binom{j}{r}=0)\\
                                            &= \frac{1}{q^{n-k}}\sum_{j=0}^n W_{\mathcal{C}^{\perp}}^j(-1)^j\binom{n-j}{r-j}(q-1)^{r-j}q^{n-r}\\
                                            &= q^{k-r}\sum_{j=0}^n(-1)^j(q-1)^{r-j}\binom{n-j}{r-j}W_{\mathcal{C}^{\perp}}^j
\end{aligned}
\]
Therefore, from \eqref{eq:macIdentity} we can derive \eqref{eq:NoOfNoneZeroR-Tuple}.
\end{itemize}
\subsection{Derive \eqref{eq:NoOfZeroR-Tuple} from \eqref{eq:macIdentity}}
By taking $r$-th partial derivative with respect to $x$ on both sides of \eqref{eq:macIdentity}, we get
\[
\begin{aligned}
 \sum_{j=0}^{n} r!\binom{n-j}{r}W_{\mathcal{C}}^jx^{n-j-r}y^j
 &= \frac{1}{q^{n-k}}\sum_{j=0}^n W_{\mathcal{C}^{\perp}}^j\sum_{i=0}^rr!\binom{n-j}{r-i}\binom{j}{i}
 \left[x+(q-1)y\right]^{n-j-r+i}(x-y)^{j-i}
\end{aligned}
\]
Let $x=y=1$, then we get
\[
\begin{aligned}
  \sum_{j=0}^{n} \binom{n-j}{r}W_{\mathcal{C}}^j&
                                            = \frac{1}{q^{n-k}}\sum_{j=0}^n W_{\mathcal{C}^{\perp}}^j\binom{n-j}{r-j}q^{n-r}\\
                                            &= q^{k-r}\sum_{j=0}^n\binom{n-j}{r-j}W_{\mathcal{C}^{\perp}}^j
\end{aligned}
\]
So from \eqref{eq:macIdentity} we can derive \eqref{eq:NoOfZeroR-Tuple}.
\subsection{Derive \eqref{eq:general} from \eqref{eq:macIdentity}}
Let $f(x,y) = W_{\mathcal{C}}(x,y)$. For $0\leq t\leq r\leq n$,
by taking $r$-th mixed partial derivatives on both sides of
$$f(x,y) = \sum\limits_{j=0}^nW_{\mathcal{C}}^jx^{n-j}y^j$$
we can get
\[
\begin{aligned}
  \frac{\partial^r f}{\partial x^{r-t}\partial y^t} &= \frac{\partial^{r-t}}{\partial x^{r-t}}\left(t!\sum_{j=0}^n\binom{j}{t}W_{\mathcal{C}}^jx^{n-j}y^{j-t}\right)\\
      &= t!\;(r-t)!\sum_{j=0}^n\binom{j}{t}\binom{n-j}{r-t}W_{\mathcal{C}}^jx^{n-j-r+t}y^{j-t}
\end{aligned}
\]
From
$$f(x,y) =\frac{1}{q^{n-k}}\sum\limits_{j=0}^nW_{\mathcal{C}^{\perp}}^j\left[x+(q-1)y\right]^{n-j}(x-y)^j$$
and Lemma~\ref{lemma:1} we get
\[
\begin{aligned}
\frac{\partial^r f}{\partial x^{r-t}\partial y^t}
    &=  \frac{1}{q^{n-k}}\;t!\;(r-t)!\sum_{j=0}^nW_{\mathcal{C}^{\perp}}^j\sum_{s=0}^{r-t}\binom{n-j}{r-t-s}\binom{j}{s}\\
    &\qquad \sum_{i=0}^{t}(-1)^i(q-1)^{t-i}\binom{n-j-r+t+s}{t-i}\binom{j-s}{i}
    \left[x+(q-1)y\right]^{n-j-r+s+i}(x-y)^{j-s-i}
\end{aligned}
\]
So we have
\[
\begin{aligned}
  \sum_{j=0}^n\binom{j}{t}\binom{n-j}{r-t}W_{\mathcal{C}}^jx^{n-j-r+t}y^{j-t}
  &= \frac{1}{q^{n-k}}\sum_{j=0}^nW_{\mathcal{C}^{\perp}}^j\sum_{s=0}^{r-t}\binom{n-j}{r-t-s}\binom{j}{s}\\
     &\qquad \sum_{i=0}^{t}(-1)^i(q-1)^{t-i}\binom{n-j-r+t+s}{t-i}\binom{j-s}{i}\\
     &\qquad\qquad \left[x+(q-1)y\right]^{n-j-r+s+i}(x-y)^{j-s-i}
\end{aligned}
\]
Substituting $1$ for $x$ and $y$, and also notice $(x-y)^{j-s-i}=0$ when $j\ne s+i$ we get
\[
\begin{aligned}
  \sum_{j=0}^n\binom{j}{t}\binom{n-j}{r-t}W_{\mathcal{C}}^j
  &= \frac{1}{q^{n-k}}\sum_{j=0}^rW_{\mathcal{C}^{\perp}}^j\sum_{i=0}^{t}\binom{n-j}{r-t-j+i}\binom{j}{j-i}
  (-1)^i(q-1)^{t-i}\binom{n-r+t-i}{t-i}q^{n-r}\\
  &= q^{k-r}\sum_{i=0}^{t}(-1)^i(q-1)^{t-i}\sum_{j=0}^r\binom{n-j}{r-j}\binom{j}{i}\binom{r-j}{t-i}W_{\mathcal{C}^{\perp}}^j
\end{aligned}
\]
So \eqref{eq:general} holds.
\subsection{Derive \eqref{eq:macIdentity} from \eqref{eq:Aj&Bj}}
Let
\[
\begin{aligned}
  f(x,y) &= W_{\mathcal{C}}(x,y) = \sum_{j=0}^nW_{\mathcal{C}}^jx^{n-j}y^{j}\\
  g(x,y) &= \frac{1}{q^{n-k}}W_{\mathcal{C}^{\perp}}\left(x + (q-1)y,x - y\right)\\
         &= \frac{1}{q^{n-k}}\sum_{j=0}^nW_{\mathcal{C}^{\perp}}^j\left[x+(q-1)y\right]^{n-j}(x-y)^j
\end{aligned}
\]
Then both $f(x,y)$ and $g(x,y)$ are homogeneous polynomials of degree $n$ in $x,y$.

For any non-negative integer $r\leq n$, by Lemma \ref{lemma:1} we have
\[
\begin{aligned}
  \left.\frac{\partial^rf}{\partial y^r}\right|_{\substack{x=1\\y=0}} &= r!W_{\mathcal{C}}^r\\
  \left.\frac{\partial^rg}{\partial y^r}\right|_{\substack{x=1\\y=0}} &= \frac{1}{q^{n-k}}\sum_{j=0}^nW_{\mathcal{C}}^j\,r!
                                                      \sum_{i=0}^r(-1)^i(q-1)^{r-i}\binom{n-j}{r-i}\binom{j}{i}
                                                      \left.\left[x+(q-1)y\right]^{n-j-r+i}(x-y)^{j-i}\right|_{\substack{x=1\\y=0}}\\
                                                      &= r!\frac{1}{q^{n-k}}\sum_{j=0}^nW_{\mathcal{C}^{\perp}}^j\sum_{i=0}^r(-1)^i(q-1)^{r-i}\binom{n-j}{r-i}\binom{j}{i}
\end{aligned}
\]
Since \eqref{eq:Aj&Bj} holds, we get
$$\left.\frac{\partial^rf}{\partial y^r}\right|_{\substack{x=1\\y=0}} = \left.\frac{\partial^rg}{\partial y^r}\right|_{\substack{x=1\\y=0}},\quad 0\leq r\leq n$$
By Lemma \ref{lemma:2} we obtain
\[
W_{\mathcal{C}}(x,y) = f(x,y)= g(x,y) = \frac{1}{q^{n-k}}W_{\mathcal{C}^{\perp}}\left(x + (q-1)y,x - y\right)
\]
\subsection{Derive \eqref{eq:macIdentity} from \eqref{eq:NoOfNoneZeroR-Tuple} or \eqref{eq:NoOfZeroR-Tuple}}
We only prove that from \eqref{eq:NoOfNoneZeroR-Tuple} we can derive \eqref{eq:macIdentity}. Let
\[
\begin{aligned}
  f(x,y) &= W_{\mathcal{C}}(x,y) = \sum_{j=0}^nW_{\mathcal{C}}^jx^{n-j}y^{j}\\
  g(x,y) &= \frac{1}{q^{n-k}}W_{\mathcal{C}^{\perp}}\left(x + (q-1)y,x - y\right)\\
         &= \frac{1}{q^{n-k}}\sum_{j=0}^nW_{\mathcal{C}^{\perp}}^j\left[x+(q-1)y\right]^{n-j}(x-y)^j
\end{aligned}
\]
Then both $f(x,y)$ and $g(x,y)$ are homogeneous polynomials of degree $n$ in $x,y$.
For any non-negative integer $r\leq n$, by Lemma~\ref{lemma:1} we get
\[
\begin{aligned}
  \left.\frac{\partial^rf}{\partial y^r}\right|_{\substack{x=1\\y=1}} &= r!\sum_{j=0}^n\binom{j}{r}W_{\mathcal{C}}^j\\
  \left.\frac{\partial^rg}{\partial y^r}\right|_{\substack{x=1\\y=1}} &= \frac{1}{q^{n-k}}\sum_{j=0}^nW_{\mathcal{C}}^j\,r!
                                                      \sum_{i=0}^r(-1)^i(q-1)^{r-i}\binom{n-j}{r-i}\binom{j}{i}
                                                      \left.\left[x+(q-1)y\right]^{n-j-r+i}(x-y)^{j-i}\right|_{\substack{x =1\\y= 1}}\\
                                                      &= r!\;q^{k-r}\sum_{j=0}^n(-1)^j(q-1)^{r-j}\binom{n-j}{r-j}W_{\mathcal{C}^{\perp}}^j
\end{aligned}
\]
From \eqref{eq:NoOfNoneZeroR-Tuple} we get
$$\left.\frac{\partial^rf}{\partial y^r}\right|_{\substack{x=1\\y=1}} = \left.\frac{\partial^rg}{\partial y^r}\right|_{\substack{x=1\\y=1}},\quad 0\leq r\leq n$$
By Lemma~\ref{lemma:2} we get $f(x,y) = g(x,y)$,
which means that
$$W_{\mathcal{C}}(x,y) = \frac{1}{q^{n-k}}W_{\mathcal{C}^{\perp}}\left(x + (q-1)y,x - y\right)$$
\subsection{Derive \eqref{eq:macIdentity} from \eqref{eq:general}}
If $t=0$ then \eqref{eq:general} reduces to \eqref{eq:NoOfZeroR-Tuple},
while if $t=r$ then \eqref{eq:general} reduces to \eqref{eq:NoOfNoneZeroR-Tuple}.
Since \eqref{eq:macIdentity} can be derived from \eqref{eq:NoOfNoneZeroR-Tuple} or \eqref{eq:NoOfZeroR-Tuple},
\eqref{eq:macIdentity} can also be derived from \eqref{eq:general}.
\section{Conclusion}
A homogeneous polynomial of degree $n$ in two variables is uniquely determined by its $n+1$ coefficients,
from the proofs in last section we can see that identities \eqref{eq:Aj&Bj}, \eqref{eq:NoOfNoneZeroR-Tuple},
\eqref{eq:NoOfZeroR-Tuple} and \eqref{eq:general} are actually four different groups of conditions that
can be used to determine the coefficients of \eqref{eq:macIdentity}, and they can be written respectively
in the following four forms:
\begin{equation}\tag{\ref{eq:Aj&Bj}$'$}\label{eq:Aj&Bjprime}
\left.\frac{\partial^{r}W_{\mathcal{C}}(x,y)}{\partial
y^r}\right|_{\substack{x=1\\y=0}} =
\left.\frac{\partial^{r}W_{\mathcal{C}^{\perp}}\left(x+(q-1)y,x-y\right)}{\partial
y^r}\right|_{\substack{x=1\\y=0}}
\end{equation}
\begin{equation}\tag{\ref{eq:NoOfNoneZeroR-Tuple}$'$}\label{eq:NoOfNoneZeroR-Tupleprime}
\left.\frac{\partial^{r}W_{\mathcal{C}}(x,y)}{\partial
y^r}\right|_{\substack{x=1\\y=1}} =
\left.\frac{\partial^{r}W_{\mathcal{C}^{\perp}}\left(x+(q-1)y,x-y\right)}{\partial
y^r}\right|_{\substack{x=1\\y=1}}
\end{equation}
\begin{equation}\tag{\ref{eq:NoOfZeroR-Tuple}$'$}\label{eq:NoOfZeroR-Tupleprime}
 \left.\frac{\partial^{r}W_{\mathcal{C}}(x,y)}{\partial x^r}\right|_{\substack{x=1\\y=1}} =
  \left.\frac{\partial^{r}W_{\mathcal{C}^{\perp}}\left(x+(q-1)y,x-y\right)}{\partial x^r}\right|_{\substack{x=1\\y=1}}
\end{equation}
\begin{equation}\tag{\ref{eq:general}$'$}\label{eq:generalprime}
\left.\frac{\partial^{r}W_{\mathcal{C}}(x,y)}{\partial
x^{r-t}\partial y^t}\right|_{\substack{x=1\\y=1}} =
\left.\frac{\partial^{r}W_{\mathcal{C}^{\perp}}\left(x+(q-1)y,x-y\right)}{\partial
x^{r-t}\partial y^t}\right|_{\substack{x=1\\y=1}}
\end{equation}
More equivalent forms of \eqref{eq:macIdentity} can be written out in this way.

\end{document}